\documentclass[12pt, twoside]{article}

\usepackage{amsmath,amsthm,amssymb}

\usepackage{times}

\usepackage{enumerate}

\makeatletter

\def\author#1{\gdef\autrun{\def\and{\unskip, }#1}\gdef\@author{#1}}

\makeatother

\usepackage{graphicx,amsthm,amssymb,amsmath,authblk}
\usepackage{hyperref,yhmath,physics,dsfont}

\newtheorem{theorem}{Theorem}
\newtheorem{lemma}[theorem]{Lemma}
\newtheorem{corollary}[theorem]{Corollary}

\newtheorem{example}[theorem]{Example}

\date{}

\title{Error-correcting codes and absolutely maximally entangled states for mixed dimensional Hilbert spaces
\thanks{16 December 2025. This research is supported by the Spanish Ministry of Science, Innovation and Universities grant PID2023-147202NB-I00.}}

\author[1]{Simeon Ball\thanks{Email: simeon.michael.ball@upc.edu.}}
\author[2]{ Raven Zhang\thanks{Email: rzhang404@gmail.com. }}

\affil[1]{Dept.~of Mathematics, Universitat Politecnica Catalunya, 08034 Barcelona}
\affil[2]{Dept.~of Physics, Universitat de Barcelona, 08028 Barcelona}

\newcommand{\one}[0]{\mathds{1}}

\begin{document}

\maketitle

\begin{abstract}
A major difficulty in quantum computation is the ability to implement fault tolerant computations, protecting information against undesired interactions with the environment. Stabiliser codes were introduced as a means to protect information when storing or applying computations in Hilbert spaces where the local dimension is fixed, i.e. in Hilbert spaces of the form $({\mathbb C}^D)^{\otimes n}$. If $D$ is a prime power then one can consider stabiliser codes over finite fields \cite{KKKS2006}, which allows a deeper mathematical structure to be used to develop stabiliser codes. However, there is no practical reason that the subsystems should have the same local dimension and in this article we introduce a stabiliser formalism for mixed dimensional Hilbert spaces, i.e. of the form ${\mathbb C}^{D_1} \otimes \cdots \otimes {\mathbb C}^{D_n}$. More generally, we define and prove a Singleton bound for quantum error-correcting codes of mixed dimensional Hilbert spaces.
We redefine entanglement measures for these Hilbert spaces and follow \cite{HESG2018} and define absolutely maximally entangled states as states which maximise this entanglement measure. We provide examples of absolutely maximally entangled states in spaces of dimensions not previously known to have absolutely maximally entangled states.
\end{abstract}

\section{Introduction}

The theory of quantum error-correction, encoding and decoding methods initiated in the latter part of the nineties with the work of Gottesman \cite{Gottesman1996}, \cite{Gottesman1997} and Calderbank et al \cite{CRSS1998}, developed quickly over the following decade, see, for example, \cite{BrunLidar2013}, \cite{FGG2007}, \cite{Gottesman2009} and \cite{LZX2008}.
The need to correct for imperfect quantum processes as well as natural decoherence of information has led to further advances, with codes such as surface codes \cite{FAMMC2012}  and bosonic codes \cite{BESWZ2024}, \cite{CMSZS2021} being proposed. These codes assume that computations are made in Hilbert spaces of constant dimension. If we assume that the local dimension of each particle is a fixed prime power then one can use stabiliser codes over finite fields \cite{KKKS2006}, which allows for a rich mathematical structure to be exploited. More generally, Scott \cite{Scott2004} assumes that the local dimension is some integer $R$, but again the local dimension is fixed. Here we consider error-correcting codes for mixed dimensional Hilbert spaces,  i.e. where each individual system can have arbitrary dimension. 

There may be physical advantages to using mixed dimensional systems. Current NISQ-era physical error rates for qubit-based computation is dominated by two-qubit gates \cite{loschnauer2025scalable}. Two-qutrit gates suffer even higher error rates by an order of magnitude \cite{saxena2025realization} than contemporary two-qubit gates due to the increase in complexity as well as tighter gaps in higher quantum energy levels. While there are no benchmarked results known to the authors for qubit-qutrit hybrid operations to date, we expect, based on the above, for restriction to that subset of operations to yield improved fidelity on physical two-qutrit gates. Thus, while constant-dimensional qudit error correction schemes exist close to the hashing bound with higher error correction thresholds than qubit-based schemes \cite{andrist2015error}, a third avenue exists in-between for closing the gap between physical error rates and error correction thresholds. This makes use of error-aware compilation methods to limit the contributions of the highest error-rate physical gates to lower dimensionality, which can only be done in mixed-dimensional schemes. 

The only other previous works dedicated to constructing quantum error correcting codes over mixed alphabets, which we are aware of, appear in \cite{GBZ2016} and \cite{WYFO2013}. In \cite{WYFO2013} two ideas are presented. The first is to take spans of graph states for mixed alphabets of size $p$ and $q$ where $p$ divides $q$ resulting in quantum error-correcting codes in spaces $(\mathbb{C}^q)^{\otimes n_1} \otimes (\mathbb{C}^{p})^{\otimes n_2}$. The second is to project from a constant dimension code onto systems defined over smaller alphabets. In \cite{GBZ2016} the main idea is to use mixed orthogonal arrays to construct quantum states of heterogenous systems. Our approach here is different from the ideas presented in \cite{GBZ2016} and \cite{WYFO2013}. We introduce a stabiliser formalism for mixed dimensional Hilbert spaces. We consider abelian subgroups $\mathcal S$ of operators which are not confined to be generalised Pauli operators, or even subgroups of some nice error basis. Previous attempts to broaden the class of stabiliser codes include considering abelian subgroups of nice error bases \cite{KR2002a} and using non-abelian subgroups \cite{KR2002b}. We will prove here that it is possible to give a general formula for the dimension of the code stabilised by any abelian subgroup $\mathcal S$, see Theorem~\ref{dimensiontracetheorem}. A very important observation here is that these new concepts and results also apply to constant dimension spaces. In this article we will be particularly interested in the class of stabiliser codes which are defined by abelian groups of unitary operators which are permutations of the computational basis elements, up to unit scalars.
We also prove a Singleton bound for such codes, generalising the bound in \cite{WYFO2013}. 
We do not address the issue of measurements for this broader class of stabiliser codes,  nor the transformations of non-Pauli unitary stabiliser groups under operations and measurements, which remains open for future work. The articles \cite{Beaudrap2013} and \cite{BLvN2016} and subsequent articles detail an easy to use formalism for simulating stabiliser circuits on arbitrary fixed local dimension \cite{Beaudrap2013}  and mixed dimensional systems \cite{BLvN2016}, generalising the classical computational efficiency of simulating the Clifford group.

The entanglement measures of Meyer and Wallach \cite{MW2002}, further developed by Brennen \cite{Brennen2003}, led to Scott \cite{Scott2004} defining an entanglement measure based on subsystem linear entropies. In the same article, Scott also proved that a maximally entangled state is equivalent to a quantum error-correcting code. In this article, we generalise Scott's entanglement measure to mixed dimensional systems, following \cite{GBZ2016}, and reprove that a maximally entangled state is equivalent to a quantum error-correcting code, see Theorem~\ref{ameisstab}. Our entanglement measure is inspired by Huber et al \cite[Section X]{HESG2018}, where they give an example of a maximally entangled state in the Hilbert space 
$\mathbb C^2 \otimes \mathbb C^3 \otimes \mathbb C^3 \otimes \mathbb C^3$. They also prove that many mixed dimensional Hilbert spaces cannot have maximally entangled states by applying the generalised shadow inequalities, introduced earlier by Rains \cite{Rains1999}, \cite{Rains2000}. In particular, $\mathbb C^2 \otimes \mathbb C^2 \otimes \mathbb C^3 \otimes \mathbb C^3$ and $\mathbb C^2 \otimes \mathbb C^2 \otimes \mathbb C^2 \otimes \mathbb C^3$ do not have absolutely maximally entangled states. It was proven in \cite{HS2000} that $\mathbb C^2 \otimes \mathbb C^2 \otimes \mathbb C^2 \otimes \mathbb C^2$ does not have absolutely maximally entangled state.

Absolutely maximally entangled (AME) states for constant dimension Hilbert spaces are equivalent to $1$-dimensional quantum maximum distance separable (MDS) codes. There are many constructions known of such codes from classical MDS codes, see for example \cite{Ball2021, HC2013}. In the articles \cite{HC2013,  HCLRL2012} it was proven that there is an equivalence of absolutely maximally entangled states shared between an even number of parties and pure state threshold quantum secret sharing schemes.

This article is organised as follows. In Section~\ref{sectionerrors}, we define detectable and correctable errors for mixed dimensional Hilbert spaces and define stabiliser codes in a more general context of any abelian group of unitary operators. In Section \ref{sectionsingleton}, we prove a Singleton bound for mixed dimensional Hilbert spaces. In Section \ref{sectionent}, we redefine Scott's entanglement measure and prove a series of lemmas which justify the use of the term entanglement measure. In Section \ref{sectioname}, we prove the equivalence between absolutely maximally entangled states and certain quantum error-correcting codes. In Section \ref{sectionexamples}, we provide examples of quantum error-correcting codes in various mixed dimensional Hilbert spaces and some examples of absolutely maximally entangled states. In Section \ref{sectionhuber}, we provide a further example of an absolutely maximally entangled state. In certain examples, we explicitly give a set of commuting unitary operators which stabilise the state.

\section{Detectable and correctable errors} \label{sectionerrors}

Let $\mathbb H={\mathbb C}^{D_1} \otimes \cdots \otimes {\mathbb C}^{D_n}$ and define any subspace $Q$ of $\mathbb H$ to be a quantum error-correcting code. Let $B$ be an orthonormal basis for $Q$. Let $\mathcal E_d$ be a set of unitary endomorphisms of $\mathbb H$. The Knill-Laflamme conditions \cite[Theorem 3.2]{KL1997} state that $\mathcal E_d$ is a set of {\em detectable errors} for $Q$ if, for any $\ket{\psi_i}, \ket{\psi_j } \in B$ and $E \in \mathcal E_d$, 
\begin{equation} \label{kl}
\bra{\psi_i} E \ket{\psi_j}= c_E \delta_{ij},
\end{equation}
where $c_E$ is constant depending only on $E$.
A set $\mathcal E_c$ is a set of {\em correctable errors} if for all $E_r, E_s \in \mathcal E_c$, we have that $E_r^{\dagger}E_s \in \mathcal E_d$. Observe that the condition implies that orthogonal states of $Q$ remain orthogonal when they are acted upon by correctable errors.

The {\em support} of a local operator $E$ on a mixed dimensional system is a subsystem $S \subset \{1,\ldots, n \}$ of minimum size such that
$$
E=E_S \otimes \one_{S'}
$$
where $S'$ is the complement of the system $S$. 

Here, the dimension of the subsystem is
$$
\dim S=\prod_{i\in S} D_i.
$$
The weight of an error is usually defined as $|S|$. The dimension of each component of $\mathbb H$ can vary, so we define the {\em dimensional weight} of $E$ as
$$
\mathrm{dimwt}(E)= \dim S,
$$
where $S$ is the support of $E$.


As in Ketkar et al \cite{KKKS2006}, we define a {\em nice error basis} to be a basis of unitary endomorphisms $\mathcal E$ containing $\one$, for which the product of any two elements of $\mathcal E$ is the scalar multiple of an element of $\mathcal E$. Furthermore, for all non-identity $E \in \mathcal E$, the operator $E$ is traceless over any subsystem of its support $S$. i.e.
$$
\tr_T(E)=0,
$$
for any subsystem for which $S \cap T$ is non-trivial.


Let $\{ \ket{j} \ | \ j \in \mathbb Z/R\mathbb Z\}$ be a basis for ${\mathbb C}^R$.

The Weyl operators on $\mathbb C^R$ are defined, for $a, b \in \mathbb Z/R\mathbb Z$, by
$$
X(a)\ket{j} =\ket{a+j}, \ \ Z(b)\ket{j}=\eta^{bj} \ket{j}
$$
where $\eta=e^{2\pi i/R}$. The operator $X(a)Z(b)$ has trace zero, unless $a=b=0$ in which case it is the identity operator.

The endomorphisms $X(a)Z(b)$ form a nice error basis for the unitary endomorphisms on ${\mathbb C}^R$, since
$$
X(a_1)Z(b_1)X(a_2)Z(b_2)=\eta^{b_1a_2} X(a_1+a_2)Z(b_1+b_2).
$$
Thus, the set of Weyl operators
$$
\mathcal E=\{ X(a_1)Z(a_1) \otimes X(a_2)Z(b_2) \otimes \cdots \otimes X(a_n)Z(b_n) \ | \ 
a_i,b_i \in (\mathbb Z/D_i \mathbb Z)\}
$$
forms a nice error basis for $\mathbb H={\mathbb C}^{D_1} \otimes \cdots \otimes {\mathbb C}^{D_n}$.

We say that a quantum error-correcting code $Q$ of $\mathbb H$ has {\em dimensional minimum distance} $D$ if $\mathcal E_d$ contains all error endomorphisms of $\mathcal E$ of dimensional weight less than $D$.
If the dimension of $Q$ is $K$ then we say that $Q$ is a 
$$
(((D_1,\ldots,D_n),K,D))
$$ 
quantum error-correcting code.
Furthermore, we say that $Q$ is a {\em pure} 
$$
(((D_1,\ldots,D_n),K,D))
$$ 
quantum error-correcting code if $c_E=0$ for all $E \neq \one$. If $K=1$ then we insist that $Q$ is pure, otherwise condition (\ref{kl}) trivially holds.

Let $\mathcal S$ be an abelian (commutative) multiplicative subgroup of unitary endomorphisms. We define a stabiliser code to be
$$
Q(\mathcal S)=\{ \ket{\psi} \in \mathbb H \ | \ M\!\ket{\psi}=\ket{\psi}\!, \ \mathrm{for} \ \mathrm{all} \ M \in \mathcal S \}.
$$
Note that although this is the usual definition of a stabiliser code \cite[Equation (2)]{KKKS2006}, the important difference here is that we do not make any restriction on the abelian subgroup $\mathcal S$ per se. It is not necessary that it should be a subgroup of the multiplicative group of Weyl operators
$$
\mathcal W_{n,R}=\{ cX(a)Z(b) \ | \ a, b \in \mathbb Z/R\mathbb Z, \ c^R=1\},
$$
or any other nice error basis \cite{KR2002a}. Furthermore, although we are motivated to introduce this generalisation because we are considering mixed dimensional Hilbert spaces, we can also apply these ideas to constant dimension Hilbert spaces.

We mention here that Ketkar et al. \cite{KKKS2006}, see also \cite{Heinrich2021}, define stabiliser codes over the multiplicative group of Pauli operators, which are defined using a finite field ${\mathbb F}_q$, where $q=p^h$ for some prime $p$. In the case of Pauli operators the $Z(b)$ operator is defined as
$$
Z(b)\ket{j}=\eta^{\mathrm{tr}(bj)} \ket{j}
$$
where $\eta$ is a primitive $p$-root of unity and $\mathrm{tr}$ is the trace map from ${\mathbb F}_q$ to ${\mathbb F}_p$. The $X(a)$ operator is again defined as
$$
X(a)\ket{j}=\ket{a+j}
$$
where addition takes place in the finite field. When $q$ is an odd prime, the Weyl operators and the Pauli operators coincide. The advantage of defining stabiliser codes over finite fields, is that one can prove that abelian subgroups of the Pauli operators correspond to symplectic subspaces of a vector space over $\mathbb F_q$. Moreover, the set of Pauli operators which are not in the $\mathrm{Centraliser}(\mathcal S) \setminus \mathcal S$, are detectable errors. If $d$ is the minimum weight of an operator in $\mathrm{Centraliser}(\mathcal S) \setminus \mathcal S$ and $|\mathcal S|=q^{n-k}$, then the code $Q(\mathcal S)$ is said to be a $[\![n,k,d]\!]_q$ code of $(\mathbb C^q)^{\otimes n}$. Moreover, if $d$ is also the minimum weight of an operator in $\mathrm{Centraliser}(\mathcal S)$ then $Q(\mathcal S)$ is pure.

Allowing any abelian subgroup of unitary operators to define a stabiliser group leads to the conclusion that any subspace $Q$ is a stabiliser code. This follows since we can define $\mathcal S$ to be the subgroup generated by the unitary operator
$$
U=2P_0-\one,
$$
where $P_0$ is the projector onto $Q$. Later, we will make some restriction on $\mathcal S$, namely, that its elements are permutations, up to multiplying by roots of unity, of the elements of the computational basis.  In other words, a typical element $M=(m_{ij})$ of $\mathcal S$ would be defined by a permutation $\sigma$ of $\{1,\ldots,\dim \mathbb H\}$, where
 $$
 m_{ij}=\delta_{j \sigma(i)} \eta_i,
 $$
 for some $\eta_i \in \mathbb C$ such that $|\eta_i|=1$. This allows for a far larger class of stabiliser codes to be found and exploited than simply restricting ourselves to a nice error basis or local Clifford operators.
 
\begin{theorem} \label{dimensiontracetheorem}
The dimension of $Q(\mathcal S)$ is 
$$
\frac{1}{|\mathcal S|} \sum_{M \in \mathcal S} \mathrm{tr}(M).
$$
\end{theorem}

\begin{proof}
Let $P$ be the projector onto the subspace $Q(\mathcal S)$. Since $P^2=P$ the eigenvalues of $P$ are zero and one. The image of $P$ is its eigenspace of eigenvalue one, which is also $Q(S)$. The eigenvalues of $P$ are zero and one, so the dimension of the eigenspace of eigenvalue one is the sum of the eigenvalues of $P$, which is equal to $\tr(P)$. 

It remains to prove that
$$
P=\frac{1}{|\mathcal S|} \sum_{M \in \mathcal S} M.
$$
Observe that 
$$
\left(\frac{1}{|\mathcal S|} \sum_{M \in \mathcal S} M\right)^2=\frac{1}{|\mathcal S|} \sum_{N \in \mathcal S} N \frac{1}{|\mathcal S|} \sum_{M \in \mathcal S} M= \sum_{N,M \in \mathcal S} \frac{1}{|\mathcal S|^2} NM=\frac{1}{|\mathcal S|} \sum_{M \in \mathcal S} M
$$
and that 
$$
\frac{1}{|\mathcal S|} \sum_{M \in \mathcal S} M
$$
fixes any element of $Q(\mathcal S)$. 

Furthermore, if 
$$
\ket{\psi}=\left(\frac{1}{|\mathcal S|} \sum_{M \in \mathcal S} M \right) \ket{\phi} 
$$
then, for all $E \in \mathcal S$,
$$
E\ket{\psi}=E\left(\frac{1}{|\mathcal S|} \sum_{M \in \mathcal S} M\right)\ket{\phi}=\frac{1}{|\mathcal S|} \sum_{M \in \mathcal S} EM \ket{\phi}=\frac{1}{|\mathcal S|} \sum_{M \in \mathcal S} M\ket{\phi}=\ket{\psi},
$$
so $\ket{\psi}$ is fixed by $\mathcal S$, i.e. $\ket{\psi} \in Q(\mathcal S)$.

Thus,
$$
P=\frac{1}{|\mathcal S|} \sum_{M \in \mathcal S} M.
$$
\end{proof}

To summarise this section, we have introduced the parameter of dimensional minimum distance for a quantum error-correcting code. We observed that the standard definition of a stabiliser code is valid for any abelian subgroup of unitary operators. We then calculated its dimension in Theorem~\ref{dimensiontracetheorem}.

\section{Quantum Singleton bound for mixed dimensional spaces} \label{sectionsingleton}

In this section we prove a quantum Singleton bound for mixed dimensional spaces. Recall that the quantum Singleton bound states that for a quantum error-correcting code $Q$ of $(\mathbb C^q)^{\otimes n}$ 
of dimension $q^k$, the minimum distance $d$ satisfies the inequality
\begin{equation} \label{qsb}
n \geqslant k+2(d-1).
\end{equation}
A code which attains this bound is called a quantum MDS code. The binary version of the quantum Singleton bound was first proved by Knill and Laflamme in \cite{KL1997} and later generalised by Rains in \cite{Rains1999b}.
Our bound is a generalisation of the Singleton bound proved in \cite[Theorem 3]{WYFO2013}. The essential difference is that we consider the dimensional minimum distance and not the usual definition of minimum distance used for constant dimensional Hilbert spaces.

To prove Theorem~\ref{singleton}, we use the following lemma, which is called the sub-additivity of entropy \cite[Equation 11.72]{NielsenChuang2000}. One can readily check that the proof of Klein's inequality \cite[Theorem 11.71]{NielsenChuang2000} and sub-additivity are valid for mixed dimensional systems.

\begin{lemma} \label{subadd}
Suppose that AB is a bipartite partition of a possibly mixed dimensional Hilbert space $\mathbb H$. Let $\rho^{AB}$ be a density operator describing the system. Then, the von-Neumann entropy satisfies
$$
S(\rho^{\mathrm{AB}}) \leqslant S(\rho^{\mathrm{A}})+ S(\rho^{\mathrm{B}})
$$
\end{lemma}

 
 Theorem~\ref{singleton} is a quantum Singleton bound for mixed dimensional spaces which generalises the quantum Singleton bound for constant dimensional spaces. 

\begin{theorem} \label{singleton}
If there is a 
$$
(((D_1,\ldots,D_n),K,D))
$$ 
quantum error-correcting code then, for any ordering of $(D_1,\ldots,D_n)$ and $r$ and $s$ such that
$$
\prod_{i=1}^s D_i <D \ \ \mathrm{and} \prod_{i=s+1}^{r+s} D_i <D,
$$
we have that
$$
\prod_{i=r+s+1}^n D_i \geqslant K.
$$
\end{theorem}

\begin{proof} 

Let $\{\ket{\psi_j} \ | \ j=1,\ldots,K\}$ be an orthonormal basis for a 
$$
(((D_1,\ldots,D_n),K,D))
$$ 
quantum error-correcting code on $\mathbb H={\mathbb C}^{D_1} \otimes \cdots \otimes {\mathbb C}^{D_n}$. We split $\mathbb H$ into three system A, B and C where
$$
\dim \mathrm A =\prod_{i=1}^s D_i <D, \ \ \dim \mathrm B=  \prod_{i=s+1}^{r+s} D_i <D
$$  
and
$$
\dim  \mathrm C=\prod_{i=r+s+1}^{n} D_i,
$$
where we choose a suitable ordering of $(D_1,\ldots,D_n)$.

Let R be a reference system of dimension $K$ with a basis $\{\ket{j} \ | \ j=1,\ldots,K\}$. 
Define a state on the RABC system,
$$
\ket{\phi}=\frac{1}{\sqrt{K}}\sum_{j=1}^K \ket{j} \ket{\psi_j}.
$$
Note that $\ket{\psi_j}$ is a basis element of the code-space belonging to the ABC system.

Let
$$
\rho^{\mathrm{RABC}}=\ketbra{\phi}.
$$

Since the code has dimensional distance $D$, any error of dimensional weight less than $D$, where the support is known, can be corrected. Therefore, we can correct errors whose support is contained in either A or B. It follows that R and A must be uncorrelated, as are R and B.

Therefore,
$$
S(\rho^{\mathrm{AR}})=S(\rho^{\mathrm A})+S(\rho^{\mathrm R}).
$$ 
and similarly for the B system.

By sub-additivity of entropy, Lemma~\ref{subadd}, we have
$$
S(\rho^{\mathrm{BC}}) \leqslant S(\rho^{\mathrm C})+S(\rho^{\mathrm B}).
$$

Since $\rho^{\mathrm{RABC}}$ is a pure state,
$$
S(\rho^{\mathrm{BC}})=S(\rho^{\mathrm{AR}}).
$$
Combining these equations, gives
$$
S(\rho^{\mathrm A})+S(\rho^{\mathrm R})\leqslant S(\rho^{\mathrm C})+S(\rho^{\mathrm B}).
$$
Switching systems A and B, we have
$$
S(\rho^{\mathrm B})+S(\rho^{\mathrm R})\leqslant S(\rho^{\mathrm C})+S(\rho^{\mathrm A}),
$$
which implies
$$
S(\rho^{\mathrm R}) \leqslant S(\rho^{\mathrm C}).
$$
Since
$$
\rho^{\mathrm R}=\frac{1}{K} \sum_{j=1}^K \ketbra{j},
$$
we have
$$
S(\rho^{\mathrm R})=\log K.
$$
Combining this with
$$
S(\rho^{\mathrm C}) \leqslant \log  \dim \mathrm C=\sum_{i=r+s+1}^n \log D_i,
$$
the claim follows.
\end{proof}

 To see that Theorem~\ref{singleton} is a generalization of the quantum Singleton bound for constant dimensional spaces, observe that a $[[n,k,d]]_q$ code, in our notation, is a 
$$
(((q,\ldots,q),q^k,q^d))
$$ 
quantum error-correcting code. Thus, we can take $s=r=d-1$ and obtain
$$
q^{n-2(d-1)} \geqslant q^k,
$$
which gives (\ref{qsb}).

\section{Entanglement measures} \label{sectionent}

Building on the ideas of Meyer and Wallach \cite{MW2002} and Brennen \cite{Brennen2003}, Scott \cite{Scott2004} introduced the following entanglement measure for states $\ket{\psi}$ in the Hilbert space $({\mathbb C}^D)^{\otimes n}$. To measure the bipartite entanglement on subsystems of $m\leqslant \frac{1}{2}n$ particles and their complements, he defined
$$
\mathrm{EM}_m (\ket{\psi})=\frac{1}{{n \choose m}} \sum_{|S|=m}\frac{D^m}{D^m-1}\left(1- \tr(\rho_S^2) \right),
$$
where $\rho_S$ is obtained from $\rho=\ketbra{\psi}$ by tracing out on the complement system $S'$, i.e. $S'=\{1,\ldots,n\} \setminus S$.

Scott proved that 
$$
0 \leqslant \mathrm{EM}_m (\ket{\psi})\leqslant 1,
$$
where
$$
\mathrm{EM}_m (\ket{\psi})=0
$$ 
if and only if $\ket{\psi}$ is a product state, i.e. $\ket{\psi}=\ket{\psi_1} \otimes \cdots \otimes \ket{\psi_n}$, where $\ket{\psi_i}$ is a state of the $i$-th subsystem. 

He also proved that $\mathrm{EM}_m(\ket{\psi})=1$ if and only if
$$
\rho_S=\frac{1}{|S|} \one_S,
$$ 
for all subsystems $S$ of at most $\frac{1}{2}n$ particles.

For our mixed dimensional Hilbert space $\mathbb H$, we define an entanglement measure, not on a specific number of systems, but on a dimension 
$$
r=\prod_{i\in S}^n D_i.
$$
for some $S \subset \{1,\ldots n\}$.

The entanglement measure we propose is then
$$
\mathrm{EM}_r (\ket{\psi})=\frac{1}{f}\sum_{\dim S=r} \frac{r}{r-1}\left(1- \tr(\rho_S^2)\right),
$$
where $f$ is the number of subsystems $S$ for which $\dim S=r$.

We now prove three essential properties which justify calling $\mathrm{EM}_r$ an entanglement measure.

\begin{lemma} \label{between0and1}
$$
0 \leqslant \mathrm{EM}_r (\ket{\psi})\leqslant 1
$$
\end{lemma}

\begin{proof}
By the spectral decomposition theorem,
$$
\rho_S=\sum_i \lambda_i \ketbra{\psi_i},
$$
where $\ket{\psi_i}$ form an orthonormal basis of eigenvectors of $\rho_S$ and $\lambda_i$ is the corresponding eigenvalue of the eigenvector $\ket{\psi_i}$.

Since $\tr(\rho_S)=1$, we have
$$
\sum_{i=1}^{\dim S} \lambda_i=1.
$$
Thus, 
$$
\tr(\rho_S^2)=\sum_{i=1}^{\dim S} \lambda_i^2.
$$
This sum is minimised when $\lambda_i=1/r$ for all $i$, where $r=\dim S$, so
$$
\tr(\rho_S^2) \geqslant r(1/r^2)=1/r.
$$
Hence,
$$
1-\tr(\rho_S^2) \leqslant 1-1/r=(r-1)/r.
$$
Thus,
$$
\mathrm{EM}_r (\ket{\psi})\leqslant 1.
$$
Since,
$$
\tr(\rho_S^2) \leqslant 1,
$$
each term in the sum in the definition of $\mathrm{EM}_r (\ket{\psi})$ is positive, which implies
$$
\mathrm{EM}_r (\ket{\psi})\geqslant 0.
$$

\end{proof}

\begin{lemma}
$$
\mathrm{EM}_r (\ket{\psi})=1
$$
if and only if 
$$
\rho_S=\frac{1}{r} \one_S
$$ 
for all subsystems $S$ where $\dim S=r$.
\end{lemma}

\begin{proof}
In the proof of Lemma~\ref{between0and1}, equality in the bound
$$
\mathrm{EM}_r (\ket{\psi})\leqslant 1
$$
occurs if and only if $\lambda_i=1/r$ for all $i$, which is if and only if
$$
\rho_S=\frac{1}{r}\sum_{i=1}^r \ketbra{\psi_i}=\frac{1}{r} \one_S.
$$
\end{proof}

\begin{lemma}
$$
\mathrm{EM}_r (\ket{\psi})=0
$$
if and only if for all subsystems $S$ where $\dim S=r$, there is a $\ket{\psi_1}$, describing a state on system $S$, and a $\ket{\psi_2}$, describing a state on the complement system $S'$, such that  
$$
\ket{\psi}=\ket{\psi_1}\otimes \ket{\psi_2}.
$$
\end{lemma}

\begin{proof}
To prove the forward implication, by Schmidt decomposition, there are orthonormal sets $\ket{\alpha_i}$ and $\ket{\beta_i}$ for $S$ and $S'$ respectively, such that
$$
\ket{\psi}=\sum_{i=1}^{\min{\{r,(\prod_{i=1}^n D_i)-r\}}} \lambda_i \ket{\alpha_i}\ket{\beta_i}.
$$
Thus,
$$
\rho=\ketbra{\psi}=\sum_{i,j} \lambda_i \overline{\lambda_j} \ket{\alpha_i}\!\ket{\beta_i}\!\bra{\alpha_j}\!\bra{\beta_j}
$$
and so
$$
\rho_S=\tr_{S'}(\rho)=\sum_k \bra{\beta_k}\left( \sum_{i,j} \lambda_i \overline{\lambda_j} \ket{\alpha_i}\!\ket{\beta_i}\!\bra{\alpha_j}\!\bra{\beta_j}
\right)\ket{\beta_k}=\sum_{k=1}^{\min{\{r,n-r\}}} \lambda_k \overline{\lambda_k} \ket{\alpha_k}\!\bra{\alpha_k}.
$$
Hence,
$$
\rho_S^2=\sum_{i=1}^{\min{\{r,n-r\}}} (\lambda_i \overline{\lambda_i})^2 \ket{\alpha_i}\!\bra{\alpha_i}
$$
and
$$
\tr(\rho_S^2)=\sum_{i=1}^{\min{\{r,n-r\}}} (\lambda_i \overline{\lambda_i})^2. 
$$
If $\mathrm{EM}_r(\ket{\psi})=0$ then we have that $\tr(\rho_S^2)=1$ for each subsystem $S$ of dimension $r$. But
$$
\tr(\rho_S^2)=1
$$
if and only if $\lambda_1=1$ and $\lambda_i=0$ for $i \geqslant 2$, which implies
$$
\ket{\psi}= \ket{\alpha_1}\ket{\beta_1}.
$$

To prove the reverse implication, suppose
$$
\ket{\psi}=\ket{\psi_1} \ket{\psi_2},
$$
where $\ket{\psi_1}$ describes a state on some subsystem $S$ of dimension $r$ and $\ket{\psi_
2}$ describes a state on the complement subsystem $S'$.

Then
$$
\rho=\ketbra{\psi}=\ket{\psi_1}\! \ket{\psi_2}\! \bra{\psi_1}\! \bra{\psi_2},
$$
so
$$
\rho_S=\ketbra{\psi_1},
$$
$$
\rho_S^2=\ketbra{\psi_1}
$$
and
$$
\tr(\rho_S^2)=1.
$$
This implies
$$
\mathrm{EM}_r(\ket{\psi})=0,
$$
since all the terms in the sum in the definition of $\mathrm{EM}_r(\ket{\psi})$ are zero.

\end{proof}

We cannot hope to prove the stronger result that $\mathrm{EM}_r (\ket{\psi})=0$ implies that $\ket{\psi}$ is a product state across all subsystems. This is easily seen by considering, for example, $\mathbb H=\mathbb C^2 \otimes \mathbb C^3 \otimes \mathbb C^5$. If we have that
$$
\mathrm{EM}_6(\ket{\psi})=0,
$$
then we can only conclude that
$$
\ket{\psi}=\ket{\psi_1}\otimes \ket{\psi_2},
$$
where $\ket{\psi_1}\in \mathbb C^2 \otimes \mathbb C^3$ can be any state.

\section{Absolutely maximally entangled states} \label{sectioname}

Following Scott \cite{Scott2004}, Goyeneche et al. \cite{GBZ2016} and Huber et al. \cite{HESG2018}, we say a state $\ket{\psi}$ is  absolutely maximally entangled if 
$$
\rho_S=\frac{1}{\dim S} \one_S
$$
for all subsystems $S$ for which 
$$
\dim S \leqslant \Delta := \sqrt{\prod_{i=1}^n D_i}.
$$
Note that for subsystems $S$ for which 
$$
\dim S> \Delta
$$
the Schmidt decomposition on systems $SS'$, where $S'$ is the complement system to $S$, implies that the rank of $\rho_S$ is at most $\dim S'$ and therefore prohibits that the marginal on $S$ be the identity operator.

More generally, we say a state $\ket{\psi}$ is maximally entangled on subsystems of dimension $r$ if 
$$
\mathrm{EM}_r (\ket{\psi})=1.
$$

For non-mixed systems, the following theorem is \cite[Proposition 3]{Scott2004}. Here, we extend this result to mixed dimensional Hilbert spaces.

\begin{theorem} \label{ameisstab}
$\ket{\psi}$ is an absolutely maximally entangled state if and only if the subspace spanned by $\ket{\psi}$ is a pure $(((D_1,\ldots,D_n),1,\lfloor \Delta \rfloor +1))$ quantum error-correcting code.
\end{theorem}

\begin{proof}
Suppose $\ket{\psi}$ is absolutely maximally entangled and let $\mathcal E$ denote a nice error basis. 

Suppose that $E \in \mathcal E$ is a non-identity operator of dimensional weight less than $\lfloor \Delta \rfloor +1$. Since the dimensional weight of $E$ is an integer it is at most $\Delta$. Let $S$ be a subsystem of dimensional weight at most $\Delta$ containing the support of $E$ and let $\ket{\phi_i}$ be an orthonormal basis for $S$. Then
$$
\bra{\psi}E\ket{\psi}=\tr(E\ketbra{\psi})
$$
$$
=\tr \left( \left(\sum_{ij} e_{ij} \ket{\phi_i}\bra{\phi_j}\otimes \one_{S'}\right)\ketbra{\psi} \right)
$$
$$
=
 \sum_{ij} e_{ij} \tr_S(\tr_{S'}((\ket{\phi_i}\bra{\phi_j} \otimes \one_{S'})\ketbra{\psi})),
$$
$$
=\frac{1}{\dim S}\sum_{ij} e_{ij} \tr_S(\ket{\phi_i}\bra{\phi_j}\one_{S})
$$
$$
=\frac{1}{\dim S} \tr_S\left(\sum_{ij} e_{ij}\ket{\phi_i}\!\!\bra{\phi_j}\right)=\frac{1}{\dim S} \tr_S(E_S)=0,
$$
where $E_S$ is the restriction of $E$ to the subsystem $S$. Thus, by (\ref{kl}), $Q$ is a $(((D_1,\ldots,D_n),1,\lfloor \Delta \rfloor +1))$ quantum error-correcting code.

Suppose that the subspace $Q$ spanned by $\ket{\psi}$ is a $(((D_1,\ldots,D_n),1,\lfloor \Delta \rfloor +1))$ quantum error-correcting code. Recall that, since $\mathcal E$ is a nice error basis, we have that the operator $E$ is traceless over any subsystem of its support $S$. i.e.
$$
\tr_T(E)=0,
$$
for any subsystem for which $S \cap T$ is non-trivial. 

We can write
$$
\ketbra{\psi}=\sum_E c_EE
$$
and so for $M \in \mathcal E$
$$
M^{\dagger}\ketbra{\psi}=\sum_E c_EM^{\dagger}E=\sum_E c_{ME} E
$$
which implies
$$
\bra{\psi}M^\dagger\ket{\psi}=\tr(M^{\dagger}\ketbra{\psi})=\tr(\one)  c_M.
$$
Thus,
$$
\ketbra{\psi}=\sum_{\mathrm{dimwt}(E) \leqslant \Delta} c_EE+\sum_{\mathrm{dimwt}(E) > \Delta} c_EE
$$
$$
=\frac{1}{\tr(\one)} \one+\sum_{\mathrm{dimwt}(E)> \Delta} c_EE,
$$
since, by (\ref{kl}), $c_E=\bra{\psi}E^\dagger\ket{\psi}=0$, if $0<\mathrm{dimwt}(E) \leqslant \Delta$.

For any subsystem $S$ of dimensional weight at most $\Delta$ and $E$ of dimensional weight more than $\Delta$, we have
$$
\tr_{S'}(E)=0,
$$
since $S'$,  the complementary system to $S$, must have a non-trivial intersection with the support of $E$. Therefore,
$$
\rho_S=\tr_{S'}(\ketbra{\psi})=\frac{\tr_{S'}(\one)}{\tr(\one)}=\frac{1}{\dim S} \one_{S}.
$$

\end{proof}

\section{Constructions of codes for mixed dimensional spaces} \label{sectionexamples}

In this section we will show, in Theorem~\ref{AMEfromStab}, a method for constructing codes for mixed or constant dimension systems from codes of systems with one component less. We then use Theorem~\ref{AMEfromStab} to construct absolutely maximally entangled states in certain mixed dimensional spaces, see Theorem~\ref{amestatefrom0}. There are known iterative constructions of absolutely maximally entangled states in constant dimension Hilbert spaces, see \cite{CW2019, HG, ZAZ2020}. The idea here is to purify a given quantum error-correcting code, as done in \cite{HG}.

\begin{theorem} \label{AMEfromStab}
If $\ket{\psi_1},\ldots, \ket{\psi_r}$ is an orthonormal set of states of a pure $(((D_1,\ldots,D_n),K, D))$ quantum error-correcting code $Q$ and $\ket{ 1},\ldots,\ket{r}$ is an orthonormal basis of $\mathbb C^r$ then
$$
\ket{\phi}=\frac{1}{\sqrt{r}}\sum_{j=1}^r \ket{j}\ket{\psi_j}
$$
spans a 
$$
(((r,D_1,\ldots,D_n),1,D))
$$ 
quantum error-correcting code for all $r\leqslant K$.
\end{theorem}

\begin{proof}
By discretisation of errors, \cite[Theorem 10.1]{NielsenChuang2000}, it suffices to prove that
$$
\bra{\phi}E \ket{\phi}=0,
$$ 
for all non-identity local operators $E=E_1 \otimes E_2$ of dimensional weight less than $D$, where $E_1$ is a Weyl operator on $\mathbb C^r$ and $E_2$ is a Weyl operator on ${\mathbb C}^{D_1} \otimes \cdots \otimes {\mathbb C}^{D_n}$. 

\underline{Case 1}
Suppose $E_2 \neq \one$. Since $Q$ is a pure $(((D_1,\ldots,D_n),K, D))$ quantum error-correcting code and the dimensional weight of $E_2$ is less than $D$, from (\ref{kl}) we have that
$$
\bra{\psi_i} E_2 \ket{\psi_j}=0,
$$
for all $i,j \in \{1,\ldots,r\}$.
Thus, 
$$
\bra{\phi}E \ket{\phi}=\frac{1}{r}\left(\sum_{i=1}^r \bra{i}\bra{\psi_i}\right)
(E_1 \otimes E_2)
\left(\sum_{j=1}^r \ket{j}\ket{\psi_j}\right)
=\frac{1}{r}\sum_{i,j=1}^r  \bra{i}E_1 \ket{j}\bra{\psi_i}E_2 \ket{\psi_j}=0.
$$

\underline{Case 2}
Suppose $E_2=\one$.  The Weyl operator $E_1=X(a)Z(b)$, for some $a,b \in \mathbb Z/r\mathbb Z$. We have that
$$
\bra{\phi}E \ket{\phi}=\bra{\phi}(X(a)Z(b) \otimes \one) \ket{\phi}=\left(\sum_{i=1}^r \bra{i}\bra{\psi_i}\right)\left(\sum_{j=1}^r \eta^{jb}\ket{j+a}\ket{\psi_j}\right)
$$
$$
=\sum_{j=1}^r \bra{\psi_{j+a}} \eta^{jb}\ket{\psi_{j}}=\sum_{j=1}^r  \eta^{jb}\bra{\psi_{j+a}}\ket{\psi_{j}},
$$
where $\eta$ is a primitive $r$-th root of unity.

If $a \neq 0$ then
$\bra{\phi}E \ket{\phi}=0$ since $\bra{\psi_{j+a}}\ket{\psi_{j}}=0$. 

If $a=0$ and $b \neq 0$ then
$\bra{\phi}E \ket{\phi}=0$ since $\sum_{i=1}^r  \eta^{ib}=0$.

\end{proof}

Recall that the quantum Singleton bound (\ref{qsb}) states that for a quantum error-correcting code $Q$ of $(\mathbb C^q)^{\otimes n}$ 
of dimension $q^k$, the minimum distance $d$ satisfies the inequality
$$
n \geqslant k+2(d-1).
$$
The Singleton bound implies, in the case that the local dimension is a fixed $q$ and $r=q$, that the dimension $K$ in Theorem~\ref{amestatefrom0} satisfies
$$
r=q \leqslant K=q^k \leqslant q^{n-2(\lceil \frac{n+1}{2} \rceil-1)},$$
which implies $n$ is odd and $K=q$.

More generally, Theorem~\ref{singleton} will give an upper bound on $K$.

The following theorem is inspired by \cite[Proposition 7]{HG}, where a similar construction was used in the case where the local dimensional is constant.

\begin{theorem} \label{amestatefrom0}
If $\ket{\psi_1},\ldots, \ket{\psi_r}$ is an orthonormal set of states of a pure 
$$
(((D_1,\ldots,D_n),K, \lceil \sqrt{rD_1\cdots D_n} \rceil ))
$$ 
quantum error-correcting code then 
$$
\ket{\phi}=\frac{1}{\sqrt{r}}\sum_{j=1}^r \ket{j}\ket{\psi_j}
$$
is an absolutely maximally entangled state of $\mathbb C^r \otimes{\mathbb C}^{D_1} \otimes \cdots \otimes {\mathbb C}^{D_n}$ for all $2 \leqslant r\leqslant K$, where $\ket{ 1},\ldots,\ket{r}$ is an orthonormal basis of $\mathbb C^r$.

\end{theorem}

\begin{proof}

Let $E=E_1 \otimes E_2$ be a local operator of dimensional weight at most 
$$
\Delta=\sqrt{rD_1\cdots D_n},$$
where $E_1$ is a Weyl operator on $\mathbb C^r$ and $E_2$ is a Weyl operator on ${\mathbb C}^{D_1} \otimes \cdots \otimes {\mathbb C}^{D_n}$. 

If $E_2$ has dimensional weight less than $\Delta$ then
$$
\bra{\psi_i} E_2 \ket{\psi_j}=0,
$$
for all $i,j \in \{1,\ldots,r\}$.
Thus, 
$$
\bra{\phi}E \ket{\phi}=0,
$$
as in the proof of Theorem~\ref{AMEfromStab}.

As in the proof of Theorem~\ref{ameisstab}, we have that
$$
\ketbra{\phi}=\frac{1}{\tr(\one)} \one+\sum_{\mathrm{dimwt}(E_2)\geqslant \Delta} c_EE.
$$

Let $S$ be a subsystem of dimension less than $\Delta$. Then $S'$, the complement system to $S$, has dimension larger than $\Delta$. Suppose $E$ is an error of the error basis such that $\mathrm{dimwt}(E_2)\geqslant \Delta$. Since $S$ has dimension strictly less than the dimensional weight of $E_2$, $S'$ must have non-trivial intersection with the support of $E_2$. Thus, $\tr_{S'}(E)=0$, for all such $E$. Hence,
$$
\rho_S(\ketbra{\phi})=\frac{1}{\dim S} \one_S.
$$

Let $S$ be a subsystem of dimension $\Delta$. (It is possible that no such subsystem exists in which case we have already proved that $\ket{\phi}$ is an absolutely maximally entangled state). Let $S'$ denote the complement system to $S$, which also has dimension $\Delta$.
By Schmidt decomposition \cite[Theorem 2.7]{NielsenChuang2000}, $\rho_S(\ketbra{\psi})$ and $\rho_{S'}(\ketbra{\psi})$ have the same set of eigenvalues. Without loss of generality, suppose $S$ contains the first subsystem. Then $S'$ must have non-trivial intersection with the support of $E_2$, for all $E=E_1 \otimes E_2$ where $\mathrm{dimwt}(E_2)\geqslant \Delta$, since 
$$
(\dim S')( \mathrm{dimwt}(E_2)) \geqslant \Delta^2 > D_1 \cdots D_n.
$$ 
Hence,
$$
\rho_S(\ketbra{\phi})=\frac{1}{\dim S} \one_S.
$$
Since $\rho_{S'}(\ketbra{\phi})$ has the same set of eigenvalues as $\rho_S(\ketbra{\phi})$, we also have that
$$
\rho_{S'}(\ketbra{\phi})=\frac{1}{\dim S'} \one_{S'}.
$$
Thus, $\ket{\phi}$ is an absolutely maximally entangled state.
\end{proof}

In the mixed dimensional notation, a $[\![n,1,\frac{1}{2}(n+1)]\!]_q$ quantum MDS code is a 
$$
(((\underbrace{q,\ldots,q}_{n \ \mathrm{times}}),q,q^{\frac{1}{2}(n+1)}))
$$ 
code. In \cite[Theorem 2]{Rains1999b}, Rains proved that a quantum MDS code is pure. 
The following corollary can be applied to these codes and we will consider some examples of this application.

\begin{corollary} \label{amestatefrom1}
If $\ket{\psi_1},\ldots, \ket{\psi_r}$ is an orthonormal set of states of a pure
$$
(((\underbrace{R,\ldots,R}_{n \ \mathrm{times}}),R,R^{\frac{1}{2}(n+1)}))
$$ 
code of $(\mathbb C^R)^{\otimes n}$ then 
$$
\ket{\phi}=\frac{1}{\sqrt{r}}\sum_{j=1}^r \ket{j}\!\ket{\psi_j}
$$
is an absolutely maximally entangled state  of $\mathbb C^r \otimes (\mathbb C^R)^{\otimes n}$ for all $r\leqslant R$, where $\ket{ 1},\ldots,\ket{r}$ is an orthonormal basis of $\mathbb C^r$.

\end{corollary}

\begin{proof}
This follows as a direct corollary to Theorem~\ref{amestatefrom0}.
\end{proof}

In the following examples we say that $Q(\mathcal S)$ is Pauli-stabilised if $\mathcal S$ is an abelian subgroup of Pauli operators. i.e. the Ketkar et al \cite{KKKS2006} definition of stabilised as mentioned in Section~\ref{sectionerrors}.

\begin{example} \label{pauliexample}
Consider the $[\![3,1,2]\!]_3$ Pauli-stabilised code of $\mathbb C^3 \otimes \mathbb C^3 \otimes \mathbb C^3$
$$
Q(\mathcal S)=\left\langle \frac{1}{\sqrt{3}}(\ket{000}+\ket{111}+\ket{222}),  \frac{1}{\sqrt{3}}(\ket{102}+\ket{210}+\ket{021}),  \frac{1}{\sqrt{3}}(\ket{201}+\ket{120}+\ket{012}) \right\rangle, 
$$
where $\mathcal S$ is the abelian subgroup generated by
$$
X(1)\otimes X(1)\otimes X(1), \ \ Z(1)\otimes Z(1)\otimes Z(1).
$$
By Corollary~\ref{amestatefrom1},  the state
$$
\ket{\phi_{23}}=\frac{1}{\sqrt{6}} (\ket{0000}+\ket{0111}+\ket{0222}+\ket{1102}+\ket{1210}+\ket{1021})
$$
is an absolutely maximally entangled state of 
$\mathbb C^2 \otimes \mathbb C^3 \otimes \mathbb C^3 \otimes \mathbb C^3$. This is easily verified by taking partial traces on any two of the qutrit systems of $\ketbra{\phi_{23}}$.

The state $\ket{\phi_{23}}$ is stabilised by the nine unitary (Pauli) operators $\one \otimes M$, where $M\in \mathcal S$. Let
$$
M_0= X \otimes (\ket{000} \mapsto \ket{102} \mapsto \ket{111} \mapsto \ket{210} \mapsto \ket{222} \mapsto \ket{021}  \mapsto )
$$ 
where $M_0$ acts as $X \otimes  \one$ on the remaining elements in the computational basis.
Then $M_0$ is a permutation operator of all the elements in the computational basis, has zero trace, order $6$ and commutes with $\one \otimes M$, for all $M\in \mathcal S$.  Thus, we can find an abelian subgroup $\mathcal S_{\mathrm{ext}}$ of size $54$ which stabilises $\ket{\phi_{23}}$, all of whose non-identity operators, have zero trace. By Theorem~\ref{dimensiontracetheorem}, $Q(\mathcal S_{\mathrm{ext}})$ is a one-dimensional code spanned by $\ket{\phi_{23}}$.

One can check directly by calculating the partial traces on any two component system that the state
$$
\ket{\phi_{33}}=\frac{1}{3} (\ket{0000}+\ket{0111}+\ket{0222}+\ket{1102}+\ket{1210}+\ket{1021}
+\ket{2201}+\ket{2120}+\ket{2012})
$$
is an absolutely maximally entangled state of 
$\mathbb C^3 \otimes \mathbb C^3 \otimes \mathbb C^3 \otimes \mathbb C^3$, as implied by Corollary~\ref{amestatefrom1}. We observe that it is Pauli-stabilised by the subgroup
generated by
$$
\one \otimes X(1)\otimes X(1)\otimes X(1), \ \ \one \otimes Z(1)\otimes Z(1)\otimes Z(1), \ \ X(1)\otimes X(2)\otimes X(1) \otimes \one,
$$
so the state $\ket{\phi_{33}}$ spans a $[[4,0,3]]_3$ stabiliser code.
\end{example}

The reason that the minimum distance increases in the example of the $[[4,0,3]]_3$ stabiliser code and in the following  $[[4,0,3]]_5$ stabiliser code follows from the argument in the proof of Theorem~\ref{amestatefrom0} where $\dim S=\Delta$. This curious behaviour, due to the Schmidt decomposition, was previously noted for constant dimensional spaces in  \cite[Proposition 7]{HG}.

\begin{example} \label{pauliexample2}
Consider the $[\![3,1,2]\!]_5$ Pauli-stabilised code of $\mathbb C^5 \otimes \mathbb C^5 \otimes \mathbb C^5$
$$
Q(\mathcal S)=\left\langle \frac{1}{\sqrt{5}}(\ket{000}+\ket{113}+\ket{221}+\ket{334}+\ket{442}), \frac{1}{\sqrt{5}}(\ket{023}+\ket{131}+\ket{244}+\ket{302}+\ket{410}), \right.
$$
$$
 \frac{1}{\sqrt{5}}(\ket{032}+\ket{140}+\ket{203}+\ket{311}+\ket{424}), 
  \frac{1}{\sqrt{5}}(\ket{041}+\ket{104}+\ket{212}+\ket{320}+\ket{433}) ,
  $$
  $$
\left.   \frac{1}{\sqrt{5}}(\ket{014}+\ket{122}+\ket{230}+\ket{343}+\ket{401}) 
\right\rangle, 
$$
where $\mathcal S$ is the abelian subgroup generated by
$$
X(1)\otimes X(1)\otimes X(3), \ \ Z(1)\otimes Z(1)\otimes Z(1).
$$
By Corollary~\ref{amestatefrom1},  the state
$$
\ket{\phi_{25}}=\frac{1}{\sqrt{10}} (\ket{0000}+\ket{0113}+\ket{0221}+\ket{0334}+\ket{0442}
+\ket{1023}+\ket{1131}+\ket{1244}+\ket{1302}+\ket{1410})
$$
is an absolutely maximally entangled state of 
$\mathbb C^2 \otimes \mathbb C^5 \otimes \mathbb C^5 \otimes \mathbb C^5$. Again, this can be easily verified by taking partial traces on any two of the ququint systems of $\ketbra{\phi_{25}}$.

The state $\ket{\phi_{25}}$ is stabilised by the $25$ unitary (Pauli) operators $\one \otimes M$, where $M\in \mathcal S$. Let
$$
M_0= X \otimes 
$$
$$
(\ket{000} \mapsto \ket{023} \mapsto \ket{113} \mapsto \ket{131} \mapsto \ket{221} \mapsto \ket{244}  \mapsto  \ket{334} \mapsto \ket{302} \mapsto \ket{442} \mapsto \ket{410}  \mapsto),
$$ 
where $M_0$ acts as $X \otimes  \one$ on the remaining elements in the computational basis.
Then $M_0$ is a permutation operator of all the elements in the computational basis, has zero trace, order $10$ and commutes with $\one \otimes M$, for all $M\in \mathcal S$.  Thus, we can find an abelian subgroup $\mathcal S_{\mathrm{ext}}$ of size $250$ which stabilises $\ket{\phi_{25}}$, all of whose non-identity operators, have zero trace. By Theorem~\ref{dimensiontracetheorem}, $Q(\mathcal S_{\mathrm{ext}})$ is a one-dimensional code spanned by $\ket{\phi_{25}}$.

Furthermore, again by Corollary~\ref{amestatefrom1},  the state
$$
\ket{\phi_{35}}=\frac{1}{\sqrt{15}} (\ket{0000}+\ket{0113}+\ket{0221}+\ket{0334}+\ket{0442}
+\ket{1023}+\ket{1131}+\ket{1244}+\ket{1302}+\ket{1410}
$$
$$
+\ket{2032}+\ket{2140}+\ket{2203}+\ket{2311}+\ket{2424}),
$$
is an absolutely maximally entangled state of $\mathbb C^3 \otimes \mathbb C^5 \otimes \mathbb C^5 \otimes \mathbb C^5$. 

and the state
$$
\ket{\phi_{45}}=\frac{1}{\sqrt{20}} (\ket{0000}+\ket{0113}+\ket{0221}+\ket{0334}+\ket{0442}
+\ket{1023}+\ket{1131}+\ket{1244}+\ket{1302}+\ket{1410}
$$
$$
+\ket{2032}+\ket{2140}+\ket{2203}+\ket{2311}+\ket{2424}
+\ket{3041}+\ket{3104}+\ket{3212}+\ket{3320}+\ket{3433}),
$$
is an absolutely maximally entangled state of $\mathbb C^4 \otimes \mathbb C^5 \otimes \mathbb C^5 \otimes \mathbb C^5$.

One can check directly by calculating the partial traces on any two component system that the state
$$
\ket{\phi_{55}}=\frac{1}{5} (\ket{0000}+\ket{0113}+\ket{0221}+\ket{0334}+\ket{0442}
+\ket{1023}+\ket{1131}+\ket{1244}+\ket{1302}+\ket{1410}
$$
$$
+\ket{2032}+\ket{2140}+\ket{2203}+\ket{2311}+\ket{2424}
+\ket{3041}+\ket{3104}+\ket{3212}+\ket{3320}+\ket{3433}
$$
$$
+\ket{4014}+\ket{4122}+\ket{4230}+\ket{4343}+\ket{4401}),
$$
is an absolutely maximally entangled state of 
$\mathbb C^5 \otimes \mathbb C^5 \otimes \mathbb C^5 \otimes \mathbb C^5$. It is Pauli-stabilised by the subgroup generated by
$$
\one \otimes X(1)\otimes X(1)\otimes X(3), \ \ \one \otimes Z(1)\otimes Z(1)\otimes Z(1), \ \ X(1)\otimes X(1)\otimes X(3) \otimes X(1),
$$
so the state $\ket{\phi_{55}}$ spans a $[[4,0,3]]_5$ stabiliser code.
\end{example}

\section{A further example} \label{sectionhuber}

In \cite{HESG2018}, Huber et al. gave an example of an absolutely maximally entangled state on $\mathbb H=\mathbb C^2 \otimes \mathbb C^3 \otimes \mathbb C^3 \otimes \mathbb C^3$. 
We give a similar example and will prove that it is an absolutely maximally entangled state by proving that the code spanned by $\ket{\psi}$ can detect all errors of dimensional weight at most $6$. Since $\Delta=\sqrt{54}$, we can apply Theorem~\ref{ameisstab} to prove that $\ket{\psi}$ is an absolutely maximally entangled state. In difference to Example~\ref{pauliexample}, one can easily verify that there are no Pauli operators stabilising the state. This can be done by verifying the there are no $X$ or $Z$ operators on the qutrit systems which fix the elements of the computational basis which have a $\ket{0}$ on the qubit system or which map these elements to elements of the computational basis which have a $\ket{1}$ on the qubit system. There is only a subgroup of order $4$ of local operators fixing the state, generated by $M_0$. However, the state is a stabilised state of an abelian group of permutations up to scalar factors. In many cases, these permutations have non-zero trace.

\begin{example} \label{notfelix}
Let
$$
\ket{\psi}=\frac{1}{\sqrt{12}}(\ket{0}(\ket{022}+\ket{201}+\ket{120}+\ket{011}+\ket{102}+\ket{210}))
$$
$$
+
\ket{1}(-\ket{101}+\ket{110}-\ket{012}+\ket{202}-\ket{220}+\ket{021}).
$$

Firstly, we prove that the subspace spanned by $\ket{\psi}$ is stabilised by an abelian subgroup $\mathcal S$ of permutations of the computational basis elements, up to multiplying by a root of unity.

Let
$$
M_0=\one \otimes \left(\begin{array}{ccc} -1 & 0 & 0 \\ 0 & i & 0 \\ 0 & 0 & i \end{array} \right)\otimes \left(\begin{array}{ccc} -1 & 0 & 0 \\ 0 & i & 0 \\ 0 & 0 & i \end{array} \right)\otimes \left(\begin{array}{ccc} -1 & 0 & 0 \\ 0 & i & 0 \\ 0 & 0 & i \end{array} \right)
$$
and observe that $M_0 \ket{\psi}=\ket{\psi}$. 

To be able to state the other stabilisers of $\ket{\psi}$ we define an orthogonal set for $\mathbb H$,
$$
\begin{array}{llll}
\ket{\phi_1}=\ket{0}\ket{022} & \ket{\phi_7}=-\ket{1}\ket{021}
 & \ket{\phi_{13}}=\ket{1}\ket{022} & \ket{\phi_{19}}=\ket{0}\ket{021}\\
\ket{\phi_2}=\ket{0}\ket{210} & \ket{\phi_8}=\ket{1}\ket{220}
& \ket{\phi_{14}}=\ket{1}\ket{210} & \ket{\phi_{20}}=-\ket{0}\ket{220}\\
\ket{\phi_3}=\ket{0}\ket{102} & \ket{\phi_9}=-\ket{1}\ket{202}
& \ket{\phi_{15}}=\ket{1}\ket{102} & \ket{\phi_{21}}=\ket{0}\ket{202}\\
\ket{\phi_4}=\ket{0}\ket{011} & \ket{\phi_{10}}=\ket{1}\ket{012}
& \ket{\phi_{16}}=\ket{1}\ket{011} & \ket{\phi_{22}}=-\ket{0}\ket{012}\\
\ket{\phi_5}=\ket{0}\ket{120} & \ket{\phi_{11}}=-\ket{1}\ket{110}
& \ket{\phi_{17}}=\ket{1}\ket{120} & \ket{\phi_{23}}=\ket{0}\ket{110}\\
\ket{\phi_6}=\ket{0}\ket{201} & \ket{\phi_{12}}=\ket{1}\ket{101}
& \ket{\phi_{18}}=\ket{1}\ket{201} & \ket{\phi_{24}}=-\ket{0}\ket{101}.\\
\end{array}
$$

To define $M_1$ we use the notation $1_2 \mapsto 2_3$ to mean that $\ket{1} $ on the second system gets mapped to $\ket{2}$ on the third system, etc. Using this notation
$$
M_1= Z \otimes (1_3 \mapsto 1_2 \mapsto 1_4 \mapsto 2_3 \mapsto 2_2 \mapsto 2_4 \mapsto) (0_3 \mapsto 0_2 \mapsto 0_4 \mapsto).
$$
As a permutation of the orthogonal set $\ket{\phi_j}$ this is the permutation (on the indices)
$$
(1\ 2\ 3\ 4\ 5\ 6\ )( 7\ 8\ 9\ 10\ 11\ 12)
(13\ -\!\!14\ 15\ -\!\!16\ 17\ -\!\!18 )(19\ -\!\!20\ 21\ -\!\!22\ 23\ -\!\!24).
$$
where $-14$ indicates $-\ket{\phi_{14}}$.

Note that $M_1$ fixes $\ket{\psi}$ since
$$
\ket{\psi}=\frac{1}{\sqrt{12}} \sum_{j=1}^{12} \ket{\phi_j}.
$$

Using the same notation, we define
$$
M_2=ZX \otimes (0_22_32_4 \mapsto 0_22_31_4 \mapsto -0_21_31_4 \mapsto 0_21_32_4 \mapsto)$$
$$
 (1_20_32_4 \mapsto 2_20_32_4 \mapsto -2_20_31_4 \mapsto 1_20_31_4 \mapsto)
(1_22_30_4 \mapsto 1_21_30_4 \mapsto -2_21_30_4 \mapsto 2_22_30_4 \mapsto)  
$$
and identity on the remaining elements in the computational basis.

As a permutation of the orthonormal set $\ket{\phi_j}$ this is the permutation (again on the indices)
$$
(1\ 7\ 4\ 10)( 2\ 8\ 5\ 11)(3\ 9\  6\ 12)(13\  19\ 16\ 22)(14\ 20\  17\ 23)(15\  21\ 18\ 24).
$$
Hence, $M_2$ also fixes $\ket{\psi}$.

Since $M_0$ is a local operator whose components are diagonal matrices it commutes with both $M_1$ and $M_2$. Furthermore,
$$
M_1M_2=M_2M_1=(1\ 8\ 6\ 7\ 5\ 12\ 4\ 11\ 3\ 10\ 2\ 9)(13\ -\!\!20 \ 18\ -\!\!19\ 17\ -\!\!24\ 16\ -\!\!23\ 15\ -\!\!22\ 14\ -\!\!21 )
$$
and $M_2$ acts as the identity on the remaining kets in the computational basis.

We can then define an abelian (commutative) subgroup
$$
\mathcal S=\langle M_0,M_1,M_2 \rangle
$$
of unitary operators on $\mathbb H$ which are permutations of the computational basis up to unit scalars.

Our next step is to prove that the subspace $Q(\mathcal S)$ of states which are stabilised by $S$ is the one-dimensional subspace spanned by $\ket{\psi}$. To do this, we calculate the trace of all the elements of $S$ and apply Theorem~\ref{dimensiontracetheorem}.

Since $M_0$ has order $4$, $M_1$ has order $6$, $M_2$ has order $4$ and $\mathcal S$ is abelian, we have that $|\mathcal S|=96$.

The 48 operators with non-zero trace are listed in the following table. Note that, since $M_1$ is the traceless $Z$ operator on the qubit system, all other operators in $\mathcal S$ will have trace zero. 

$$
\begin{array}{|c|c|c|c|c|c|c|c|}
\hline
 & \tr &  & \tr &  & \tr &  & \tr \\ \hline
 \one & 54 & M_0 & 2(2i-1)^3 & M_2 & 30 & M_0M_2 & -4i-2\\
M_0^2 & -2 & M_0^3 & 2(-2i-1)^3 & M_0^2M_2& -26 & M_0^3M_2 & 4i-2\\
M_1^2 & 6 & M_0M_1^2 & -4i-2 & M_1^2M_2 & 6 & M_0M_1^2M_2 & -4i-2\\
M_0^2M_1^2 & -2 & M_0^3M_1^2 & 4i-2 & M_0^2M_1^2M_2 & -2 & M_0^3M_1^2M_2 & 4i-2\\
M_1^4 & 6 & M_0 M_1^4& -4i-2 & M_1^4M_2 & 6 & M_0M_1^4M_2 & -4i-2\\
M_0^2M_1^4  & -2 & M_0^3M_1^4  & 4i-2 & M_0^2M_1^4 M_2 & -2 & M_0^3M_1^4 M_2 & 4i-2\\
M_2^2 & 30 & M_0M_2^2 & -4i-2 & M_2^3 & 30 & M_0M_2^3 & -4i-2\\
M_0^2M_2^2 & -26 & M_0^3M_2^2 & 4i-2 & M_0^2M_2^3& -26 & M_0^3M_2^3 & 4i-2\\
M_1^2M_2^2 & 6 & M_0M_1^2M_2^2 & -4i-2 & M_1^2M_2^3 & 6 & M_0M_1^2M_2^3 & -4i-2\\
M_0^2M_1^2M_2^2  & -2 & M_0^3M_1^2M_2^2 &4i-2 & M_0^2M_1^2M_2^3 & -2 & M_0^3M_1^2M_2^3& 4i-2\\
M_1^4M_2^2  & 6 & M_0M_1^4M_2^2 & -4i-2 & M_1^4M_2^3& 6 & M_0M_1^4M_2^3 & -4i-2\\
M_0^2M_1^4M_2^2  & -2 & M_0^3M_1^4M_2^2 & 4i-2 & M_0^2M_1^4M_2^3 & -2 & M_0^3M_1^4M_2^3& 4i-2\\ \hline
   \end{array}
$$

By Theorem~\ref{dimensiontracetheorem}, the dimension of $Q(\mathcal S)$ is (summing the sums of the four columns)
$$
(72+24+24-24)/96=1
$$
since 
$$
2(2i-1)^3+2(-2i-1)^3=44.
$$

Since we already observed that $\ket{\psi}$ is stabilised by the elements of $\mathcal S$, we conclude that
$$
Q(\mathcal S)= \langle \ket{\psi} \rangle,
$$
where
$$
\ket{\psi}=\frac{1}{\sqrt{12}}(\ket{0}(\ket{022}+\ket{201}+\ket{120}+\ket{011}+\ket{102}+\ket{210}))
$$
$$
+
\ket{1}(-\ket{101}+\ket{110}-\ket{012}+\ket{202}-\ket{220}+\ket{021}).
$$

It remains to prove that
$$
\bra{\psi}E \ket{\psi}=0
$$
for all operators $E$ of dimensional weight less than $\sqrt{54}$. 

By symmetry on the second, third and fourth particles, it suffices to consider
$$
E \in \{ \sigma_1 \otimes \one \otimes \one \otimes \one  , \
\one   \otimes \sigma_2\otimes \one \otimes \one , \
\sigma_1 \otimes \sigma_2 \otimes \one \otimes \one   \},
$$
where $\sigma_1$ is a Pauli operator on the qubit system and $\sigma_2$ is a Pauli operator on the first qutrit system.

Since the error operators are also permutations of the computational basis elements, up to unit scalars, it suffices to calculate the intersection of the image of $E$ on the support of $\ket{\psi}$, with the support of $\ket{\psi}$.
Using the same notation as above, the support of $\ket{\psi}$ is $(1,\ldots,12)$ and we can calculate the following table.

\begin{center}
\begin{tabular}{|c|c|c|} \hline
$E$ & intersection with $(1,\ldots,12)$ with  & $\bra{\psi}E \ket{\psi}$   \\ 
 &  image of $E$ on support of $\ket{\psi}$ &    \\ \hline
$X \otimes \one \otimes \one \otimes \one$ & $\emptyset$ & $0$ \\ \hline
$Z \otimes \one \otimes \one \otimes \one$ & $(1,\ldots,6,-7,\ldots,-12)$ & $\frac{1}{12}(6-6)=0$ \\ \hline
$ZX \otimes \one \otimes \one \otimes \one$  &  $\emptyset$  & $0$ \\ \hline
$\one \otimes  Z(b)  \otimes\one \otimes \one$  & $(1,\eta^{2b} 2, \eta^b 3, 4, \eta^b 5,\eta^{2b}6,$ & \\
&  $ 7, \eta^{2b}8,\eta^{2b}9,10,\eta^b 11, \eta^b 12)$ & $4(1+\eta+\eta^2)=0$ \\ \hline
$X \otimes Z(b)  \otimes\one \otimes \one$ &   $\emptyset$  & $0$ \\ \hline
$Z \otimes Z(b)  \otimes\one \otimes \one$ &  $(1,\eta^{2b} 2, \eta^b 3, 4, \eta^b 5,\eta^{2b}6,$ & \\
& $- 7, -\eta^{2b}8,-\eta^{2b}9,-10, -\eta^b 11, -\eta^b 12)$ & $(2-2)(1+\eta+\eta^2)=0$ \\ \hline
$ZX \otimes Z(b)  \otimes\one \otimes \one $ &  $\emptyset$  & $0$  \\ \hline

$\one \otimes  X(1)Z(b)  \otimes\one \otimes \one $ &  $\emptyset$  & $0$ \\ \hline
$X \otimes X(1)Z(b)  \otimes\one \otimes \one$ & $(-\eta^b 2, \eta^b 6,\eta^b 8,-\eta^b 9  )$ & $\frac{1}{12}(2\eta^b-2\eta^b)=0$ \\ \hline
$Z \otimes X(1)Z(b)  \otimes\one \otimes \one $ &  $\emptyset$ & $0$ \\ \hline
$ZX \otimes X(1)Z(b)  \otimes\one \otimes \one $ & $(-\eta^b 2, \eta^b 6,-\eta^b 8,\eta^b 9 )$& $\frac{1}{12}(2\eta^b-2\eta^b)=0$  \\ \hline
$\one \otimes  X(2)Z(b)  \otimes\one \otimes \one $ &  $\emptyset$ & $0$ \\ \hline
$X \otimes X(2)Z(b)  \otimes\one \otimes \one $ & $(-\eta^{2b} 3, \eta^{2b} 5,-\eta^{2b} 11, \eta^{2b} 12)$ & $\frac{1}{12}(2\eta^{2b}-2\eta^{2b})=0$ \\ \hline
$Z \otimes X(2)Z(b)  \otimes\one \otimes \one $ & $\emptyset$ & $0$ \\ \hline
$ZX \otimes X(2)Z(b)  \otimes\one \otimes \one $ & $(-\eta^{2b} 3,\eta^{2b} 5,\eta^{2b} 11, -\eta^{2b} 12  )$ & $\frac{1}{12}(2\eta^{2b}-2\eta^{2b})=0$  \\ \hline
\end{tabular}
\end{center}

Thus, $Q(\mathcal S)$ is a $(((2,3,3,3),1,\sqrt{54}))$ quantum error correcting code and, by Theorem~\ref{ameisstab}, $\ket{\psi}$ is an absolutely maximally entangled state.
\end{example}

Consider the subgroups of index two of $\mathcal S$ from Example~\ref{notfelix}. We might hope that these might provide higher dimensional codes which can also detect dimensional weight errors of weight less than $\Delta$. However, in each case we find a two-dimensional undetectable error.

\begin{center}
\begin{tabular}{|c|c|c|c|} \hline
name & generators & basis for $Q(\mathcal S_i)$ & undetectable error \\ \hline
$\mathcal S_0$ & $M_0^2$, $M_1$, $M_2$ & $\ket{\psi}$, $\ket{0000}$ & $Z\otimes \one \otimes \one \otimes\one$ \\ \hline
$\mathcal S_1$ & $M_0$, $M_1^2$, $M_2$ & $\ket{\psi}$, $\frac{1}{\sqrt{12}} \sum_{j=13}^{24} \ket{\phi_j}$ & $ZX\otimes \one \otimes \one \otimes\one$ \\ \hline
$\mathcal S_2$ & $M_0$, $M_1$, $M_2^2$ & $\frac{1}{\sqrt{6}} \sum_{j=1}^{6} \ket{\phi_j}$,\ $\frac{1}{\sqrt{6}} \sum_{j=7}^{12} \ket{\phi_j}$ &  $Z\otimes \one \otimes \one \otimes\one$ \\ \hline
\end{tabular}
\end{center}

There are interesting related subgroups which also yield absolutely maximally entangled stabilised states. For example, the subgroup $\langle M_0,-M_1,M_2 \rangle$ stabilises the absolutely maximally entangled state
$$
\frac{1}{\sqrt{12}} \sum_{j=13}^{24} \ket{\phi_j}.
$$

\section{Further comments}
In this article we have proved a quantum Singleton bound for quantum error correcting codes for mixed dimensional Hilbert spaces. 
Ketkar et al. \cite{KKKS2006} proved a quantum Hamming bound for constant dimensional spaces from which one should be able to generalise to mixed dimensional spaces.

We have provided many examples of absolutely maximally entangled states in mixed dimensional Hlbert spaces. A natural question to ask is if there exists an orthonormal basis of absolutely maximally entangled states for the space. For example, the Bell states form an orthonormal basis of absolutely maximally entangled states of $\mathbb C^2 \otimes \mathbb C^2$.

The stabiliser codes for mixed systems we have found are stabilised by abelian groups of unitary operators which are permutations, up to unit scalars, of the computational basis elements. These operators in many cases are non-local. It seems unlikely that these codes are stabilised by local operators since for a mixed system including subsystems of dimension $D_1$ and $D_2$ where $D_1$ and $D_2$ are co-prime, it is only possible to have two such local stabilizers commute if they independently commute
over the qudits of each dimension. Observe that a pair of Weyl operators
of dimension $D$ commute up to a phase of $\eta^j$ for some integer $j$, where $\eta$ is a primitive $D$-th root of unity. For example, with a single qubit subsystem, this would restrict the qubit part of such local operators to only be the identity, and codes constructed from such operators would fail to detect any errors on the qubit.

It would be interesting to have more examples of codes for mixed dimensional subspaces with larger dimension $K$. Especially those which attain the quantum Singleton bound for mixed dimensional spaces given by Theorem~\ref{singleton}.

\section{Acknowledgements}
We would like to thank Felix Huber for his very helpful comments and the anonymous referees for their comments.

\vspace{1cm}

   Simeon Ball\\
   Departament de Matem\`atiques, \\
Universitat Polit\`ecnica de Catalunya, \\
Modul C3, Campus Nord,\\
Carrer Jordi Girona, 1-3\\
08034 Barcelona, Spain \\
   {\tt simeon.michael.ball@upc.edu} \
   
   \bigskip
   
     Raven Zhang\\
   Departament de F\'isica, \\
Universitat de Barcelona, \\
Carrer Mart\'i i Franqu\`es, 1-11\\
08028 Barcelona, Spain \\
   {\tt  rzhang404@gmail.com} \

\end{document}